\documentclass[journal]{IEEEtran}
\usepackage{flushend}
\usepackage{times}
\usepackage{cite}
\usepackage{color}
\usepackage{epsf}
\usepackage{epsfig}
\usepackage{graphicx}
\usepackage{graphics}
\usepackage{epstopdf}
\usepackage[footnotesize]{caption}
\usepackage{amsmath}
\usepackage{amssymb}
\usepackage{amsthm}
\usepackage{amsxtra}
\usepackage[ruled]{algorithm2e}
\usepackage{algorithmic}
\usepackage{enumerate}
\usepackage{multirow}
\usepackage{enumerate}
\usepackage{amssymb}
\usepackage{lipsum}
\usepackage{setspace}
\usepackage{multirow}

\newtheorem{theorem}{\bf Theorem}
\newtheorem{proposition}{\bf Proposition}

\newtheorem{definition}{\bf Definition}
\newtheorem{remark}{\bf Remark}

\usepackage{rotating} 
\usepackage{graphicx} 
\usepackage{draftwatermark}
\SetWatermarkText{IEEE Trans. Smart grid-2015}
\SetWatermarkScale{0.35}

\newcommand{\Rmnum}[1]{\expandafter\@slowromancap\romannumeral #1@}

\makeatother
\begin{document}
\title{Price Discrimination for Energy Trading in Smart Grid: A Game Theoretic Approach}
\author{Wayes~Tushar,~\IEEEmembership{Member,~IEEE,}~Chau~Yuen,~\IEEEmembership{Senior Member,~IEEE,}~David B. Smith,~\IEEEmembership{Member, IEEE,} and~H.~Vincent~Poor,~\IEEEmembership{Fellow,~IEEE}
\thanks{W. Tushar and C. Yuen are with Singapore University of Technology and Design (SUTD), 8 Somapah Road, Singapore 487372. (Email: \{wayes\_tushar, yuenchau\}@sutd.edu.sg).}
\thanks{D. B. Smith is with the National ICT Australia (NICTA), Australian Technology Park, Eveleigh, NSW 2015, Australia. (Email: david.smith@nicta.com.au).}
\thanks{H. V. Poor is with the School of Engineering and Applied Science at Princeton University, Princeton, NJ 08544, USA. (Email: poor@princeton.edu).}
\thanks{This work is supported in part by the Singapore University of Technology and Design (SUTD) through the Energy Innovation Research Program (EIRP) Singapore NRF2012EWT-EIRP002-045, and in part by the U.S. National Science Foundation under Grant ECCS-1549881.}
\thanks{D. B. Smith's work is supported by NICTA. NICTA is funded by the Australian Government through the Department of Communications and the Australian Research Council through the ICT Centre of Excellence Program.}
}
\IEEEoverridecommandlockouts
\maketitle
\begin{abstract}
Pricing schemes are an important smart grid feature to affect typical energy usage behavior of energy users (EUs). However, most existing schemes use the assumption that a buyer pays the same price per unit of energy to all suppliers at any particular time when energy is bought. By contrast, here a discriminate pricing technique using game theory is studied. A cake cutting game is investigated, in which participating EUs in a smart community decide on the price per unit of energy to charge a shared facility controller (SFC) in order to sell surplus energy. The focus is to study fairness criteria to maximize sum benefits to EUs and ensure an envy-free energy trading market. A benefit function is designed that leverages generation of discriminate pricing by each EU, according to the amount of surplus energy that an EU trades with the SFC and the EU's sensitivity to price. It is shown that the game possesses a socially optimal, and hence also Pareto optimal, solution. Further, an algorithm that can be implemented by each EU in a distributed manner to reach the optimal solution is proposed. Numerical case studies are given that demonstrate beneficial properties of the scheme.
\end{abstract}
 \setcounter{page}{1}
\section{Introduction}\label{sec:introduction}
\IEEEPARstart{O}{ne} of the main stimuli behind adopting energy management in smart grid is the use of different pricing schemes, in which an energy entity changes the price of per unit electricity according to the generation and demand so as to motivate users to modify their attitudes towards electricity consumption and supply~\cite{Fang-J-CST:2012,YangCao:2012,LiangYu:2015}. Particularly, with the advancement of distributed energy resources (DERs), different pricing techniques can assist the grid or other energy entities, such as shared facility controllers (SFCs)\footnote{An SFC, as we will see in the next section, is an entity that is responsible for managing energy of different shared facilities of a smart community.} to operate reliably and efficiently by obtaining some energy supply from the energy users (EUs)~\cite{Tushar-globecom:2014}.

Over the past few years there has been significant interest in devising pricing schemes for energy management in smart grid. These schemes can be classified into three general categories: time-of-use pricing; day-ahead pricing; and real-time pricing~\cite{Dong-TSG:2013}. Time-of-use pricing has three different pricing rates: peak, off-peak and shoulder rates based on the time of electricity use by the EUs~\cite{Asano-TPS:1992}. Day-ahead pricing is determined by matching offers from generators to bids from EUs in order to develop a classic supply and demand equilibrium price at an hourly interval~\cite{Torre-TPS:2002}. Finally, real-time pricing~\cite{Dong-TSG:2013} refers to tariffed retail charges for delivering electric power and energy that vary over hour-to-hour, and are determined through an approved methodology from wholesale market prices. Other pricing schemes that have been discussed in the literature include critical peak pricing, extreme day pricing, and extreme day peak pricing~\cite{Dong-TSG:2013}. Discussion of various pricing schemes for energy management in smart grid can be found in \cite{Jiang-TSG:2014,Ma-TCST:2014,Yoon-TSG:2014,Can-Wan-TSG:2014,Peng_wang-TSG:2014,M-Rahimiyan-TPS:2014,Zhi-yu-Xu-TPS:2014,Tushar-TSG:2013} and the references therein. Nevertheless, an important similarity in most of these pricing schemes is that all of the EUs decide on the same selling price per unit of energy at a particular time.

With the increase in government subsidies for encouraging the use of renewables, more EUs with DERs are expected to be available in the future~\cite{Liu-STSP:2014,Naveed-Energies:2013,YiLiu-TIE:2015,Naveed-TSG:2015,Naveed-Elsevier:2016}. This will subsequently lead to a better completion of purchasing targets for SFCs in order to maintain electricity for shared facilities in a community~\cite{Tushar-TIE:2014}. This is due to the fact that the opportunity of an SFC to trade electricity with EUs can greatly reduce its dependency on the grid, and consequently decrease the cost of energy purchase. However, not all EUs would be interested in trading their surplus energy if the benefit from the SFC is not attractive~\cite{Tushar-globecom:2014}. In particular, as we will see shortly, this can happen to EUs with limited energy surplus and/or with higher sensitivity to price whose respective return could be very small. Nevertheless, as shown in \cite{Tushar-globecom:2014}, one possible way to address this problem is that these EUs can sell their energy at a relatively higher price per unit of energy, within a reasonable margin, compared to EUs with very large DERs (and/or, with lower sensitivity to the choice of price) without affecting their revenue significantly.

It is natural to think that the benefit to the end-user will increase if the price for selling each unit of energy increases. However, we note that in an energy market with a large number of sellers the buyer has many choices to buy the electricity. Hence, a significantly higher price per unit of energy from a seller can motivate the buyer to switch its buying from that expensive option to a seller who is selling at a comparatively cheaper rate~\cite{Wayes-J-TSG:2012}. Thus, even with a higher selling price per unit of energy, the net benefit to a user may decrease significantly if the amount of energy that it can sell to the buyers, i.e., the SFC in this case, becomes very small. This type of phenomenon has occurred recently in the global oil market~\cite{oil_price:2015}. This can further be illustrated by a toy example as follows.

Consider the numerical example given in Table \ref{table:motivation} where EU1 and EU2 sell their surplus energy to the SFC to meet the SFC's $40$ kWh energy requirement. It is considered that EU1 and EU2 have DERs with capacity of $50$ kWh and $10$ kWh respectively (and EU1 is significantly larger than EU2 in terms of available energy to supply). In case 1, EU1 and EU2 sell $35$ and $5$ kWh of energy to the SFC at a price of $20$ cents/kWh. Hence, the revenues of EU1 and EU2 are $700$ and $100$ cents respectively and the total cost to the SFC is $800$ cents. In case 2, EU1 and EU2 choose their prices differently and sell their surplus energy at rates of $18$ cents/kWh and $22$ cents/kWh respectively. Now, due to the change of price in case 2, if EU1 reduces its selling amount to $32$ (since the price is reduced) and EU2 increases its surplus amount to $8$ kWh (as the current price is high) the resulting revenue changes to $576$ and $176$ cents respectively for EU1 and EU2 whereas the total cost to the SFC reduces to $752$ cents. Thus, according to this example, it can be argued that discriminate pricing is  considerably beneficial to EUs with small energy (revenue increment is $76\%$) at the expense of  relatively lower revenue degradation (e.g., $17\%$ in case of EU1) from EUs with larger energy capacity. It further reduces the cost to the SFC by $6\%$.

However, one main challenge for such price discrimination among different EUs, which is not discussed in \cite{Tushar-globecom:2014} and yet needs to be explored, is the maintenance of fairness of price distribution between different EUs to enable such schemes to be sustained in electricity trading markets. For example, if the EUs are not happy with the price per unit of energy that they use to sell their surplus energy, or if they envy each other for the adopted discrimination, energy markets that practice such discriminate pricing schemes would eventually diminish. Hence, there is a need for solutions that can maintain the price disparity between EUs in a fair manner, whereby considering their available surplus energy and their sensitivity to change of price, an envy-free energy trading environment with a view to obtain a socially optimal\footnote{A socially optimal solution maximizes the sum of benefits to all EUs in the smart grid network~\cite{Yudovina:2014}.} energy management solution is ensured.

To this end, \emph{this paper complements the existing pricing schemes in the literature by studying the fairness of selecting different prices for different EUs in smart grid.} However, unlike \cite{Tushar-globecom:2014}, where a two-stage Stackelberg game is studied, we take a different approach in this paper. Particularly, we explore a \textbf{cake-cutting game (CCG)} for selecting discriminate prices for different users. In the proposed CCG, the EUs with smaller available energy can decide on a higher unit price, and the price is also adaptive to the sensitivity of EUs to the choice of price. A suitable benefit function is chosen for each of the EUs that enables the generation of discriminate pricing so as to achieve a \emph{socially optimal} solution of the game. Thus, also, Pareto optimality is directly implied by this socially optimal solution, and hence an envy-free energy trading market is established. We propose an algorithm that can be adopted by each EU in a distributed manner to decide on the price vector by communicating with the SFC, and the convergence of the algorithm to the optimal solution is demonstrated. Finally, we present some numerical case studies to show the beneficial properties of the proposed discriminate pricing scheme.

We stress that the current grid system does not allow such discriminate pricing among EUs. Nonetheless, the idea of price discrimination is not new and has been extensively used in economic theory. For instance, the effect of price discrimination on social welfare is first investigated in 1933~\cite{Robinson:1933}, which is further extended with new results in \cite{Schmalensee:1981} and \cite{Varian:1985}. In recent years, the authors in \cite{Achim:2014} study the airport congestion pricing technique when the airline carriers discriminate with respect to price.  In \cite{Grennan:2012}, the authors use a new panel of data on buyer-supplier transfers and build a structural model to empirically analyze bargaining and price discrimination in a medical device market. The study of intertemporal price discrimination between consumers who can store goods for future consumption needs is presented in \cite{Hendel:2013}, and the effects of third-degree price discrimination on aggregate consumers surplus is considered in \cite{Cowan:2012}. Further, a framework for flexible use of cloud resources through profit maximization and price discrimination is studied in \cite{Tsakalozos:2011}. In this context, we also envision discriminate pricing as a further addition to real-time pricing schemes in future smart grid. Such a scheme is particularly suitable for the energy trading market when the SFC may want to reduce its dependence on a single dominant user. For example, in the toy example the SFC may rely heavily on User $1$ who has a large surplus for the same price model. However, by giving more incentive to User $2$, the SFC managed to reduce its dependence on User $1$ by buying less energy from it. Please note that this could happen in a real-world scenario in which a buyer pays a small supplier a relatively higher price in order to help the small supplier grow, and at the same time to reduce the dependence on a single big supplier. Such trading will prevent the possibility of the big supplier growing too big and creating a monopoly, which could lead to a serious problem in the long run. Please note that examples of such differentiation can also be found in current standard Feed-in-Tariff (FIT) schemes~\cite{2012-solarchoice}.

\begin{table}[t!]
\caption {Numerical example of a discriminate pricing scheme where an SFC requires $40$ kWh of energy from two EUs and the SFC's total price per unit of energy to pay to the EUs is $40$ cents/kWh.}
\centering
\begin{tabular}{|c||c|c|}
\hline
& \textbf{Case 1} & \textbf{Case 2}\\
\hline
Payment to EU1 (cents/kWh) & $20$ & $18$\\
\hline
Payment to EU2 (cents/kWh)& $20$ & $22$\\
\hline
Energy supplied by EU1 (kWh) & $35$ & $32$\\
\hline
Energy supplied by EU2 (kWh) & $5$ & $8$\\
\hline\hline
\textbf{Revenue of EU 1} (cents) & $700$ & $576$ (-$17\%$)\\
\hline
\textbf{Revenue of EU 2} (cents) & $100$ & $176$ (+$76\%$)\\
\hline
\textbf{Cost to the SFC} (cents) & $800$ & $752$ (-$6\%$)\\
\hline
\end{tabular}
\label{table:motivation}
\end{table}

The remainder of the paper is organized as follows. The system model is described in Section \ref{sec:system-model} and the problem of price discrimination is studied as a CCG in Section \ref{sec:game}. We provide numerical case studies to show the beneficial properties of the proposed scheme in Section \ref{sec:Case Study}. Finally, we draw some concluding remarks in Section \ref{sec:Conclusion}.
\section{System Description}\label{sec:system-model}
\begin{figure}[t]
\centering
\includegraphics[width=\columnwidth]{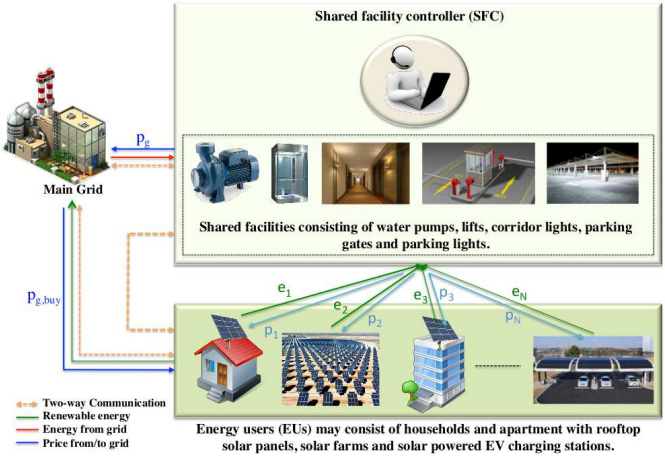}
\caption{Demonstration of the elements of the considered system model and the direction of the flow of power and information within the system.} \label{fig:SystemModel}
\end{figure}
Consider a smart grid system consisting of $N$ EUs, where $N = |\mathcal{N}|$, an SFC and a main grid. Each EU $n\in\mathcal{N}$ can be considered as a single user or a group of users connected via an aggregator that acts as a single entity. Each EU is equipped with DERs such as solar panels and wind turbines and does not have any storage facility\footnote{An example of such a system is a grid-tie solar system without storage device~\cite{GridTie:2015}.}. Hence, EU $n$ needs to sell its surplus energy, if there is any, either to the SFC or to the main grid after meeting its essential demand in order to make some extra revenue. Due to the fact that the buying price $p_{g,\text{buy}}$ of a main grid is considerably lower than pricing within a facility~\cite{McKenna-JIET:2013}, it is reasonable to assume that the EUs would be more keen to sell their surplus to the SFC instead of to the grid\footnote{For example, in the state of Queensland in Australia, the selling price of electricity is $16.262$ cents/kWh during off peak hours (which is almost double during peak hours) \cite{Queensland:2015}, whereas the buying price of electricity is $6.348$ cents/kWh under Queensland's Feed-in Tariff scheme \cite{QueenslandFeed-in-Tariff:2015}.}. Alternatively, if the payment from the SFC is not sufficiently attractive to any EU, the EU may schedule its equipment for extra consumption or may choose to sell to the grid instead of selling to the SFC. The SFC, on the other hand, controls the electricity consumption of equipment and machines that are shared and used by the EUs on a regular basis, and does not have any energy generation capacity. Therefore, the SFC relies on EUs and the grid for its required electricity. Essentially, the SFC is interested in buying  as much energy as possible from the EUs as the buying price per unit of energy from the main grid is significantly higher. All the EUs, the main grid and the SFC are connected to each other through power and communication lines~\cite{Fang-J-CST:2012}. A schematic representation of the considered system model is shown in Fig.~\ref{fig:SystemModel}.

To this end, let us assume that at a particular time during a day each EU $n\in\mathcal{N}$ has an energy surplus of $e_n$ that it wants to sell either to the SFC or to the grid with a view to make extra revenue. Each EU $n$ charges the SFC a price $p_n$ per unit of energy for selling its $e_n$. $p_n,~\forall n$ do not need to be equal to each other and can be varied according to the amount of surplus at each EU and the EU's sensitivity $\alpha_n>0$ to the choice of price. The SFC, on the other hand, wants to buy this energy $e_n$ from each EU $n$ in order to meet the demands of the shared facilities of the community. We assume that the SFC has a budget constraint $C$, and hence, the sum of what the SFC needs to pay to all EUs needs to satisfy,
\begin{eqnarray}
\sum_n e_n p_n\leq C.\label{eqn:budget-price}
\end{eqnarray}
Please note that such a budget is necessary to prevent the EUs from increasing their selling price per unit of energy considerably and thus maintaining market competitiveness, which may arise due to allowing EUs to decide on the selling price through interacting with the SFC\footnote{As we will see later, EUs and the SFC interact with each other in the proposed scheme to decide on the price vector.}. This budget also enables us to decouple the SFC's decision making process of buying energy from the EUs (which is the main focus of this work) from the problem of the SFC's buying energy from the grid. However, the budget, which facilitates price discrimination in the proposed scheme, needs to be chosen such that it is always lower than the total price that the SFC needs to spend buying the same amount of energy from the grid and thus always benefits the SFC in terms of reducing cost. This is due to the fact that if the total money that the SFC pays to the EUs becomes equal to (or greater than) the amount that it needs to pay the grid, the SFC will not be encouraged to buy energy from the EUs as the SFC can buy all its required energy from the grid independently. Nevertheless, a suitable value for a budget $C$ may depends on many other factors, e.g., how willing is a user to sell its surplus to the SFC, which requires human behavioral models to obtain reasonable estimates. For example, as a rational entity, each EU $n$ would be interested in charging the SFC as much price per unit of energy as possible for selling its surplus. However, on the one hand, a very high price may discourage the SFC from trading any energy with the EU and rather motivate the SFC to buy its energy from the grid. On the other hand, if the selling price of an EU is too small, this may compel the EU to withdraw its surplus from the energy market as the expected revenue from energy trading with the SFC would be significantly lower. Nonetheless, in this paper we consider a general setting, i.e., a budget $C$ for the SFC when it buys energy from the EUs.

Now, the choice of price $p_n$ by an EU $n$ is also restricted by its sensitivity to the choice of price. For example, as discussed in Table~\ref{table:motivation}, a smaller price may not affect an EU with a very large surplus (i.e., lower sensitivity), but it can significantly alter the total revenue of an EU with lower available energy (higher sensitivity). In this context, it is considered that the price $p_n$ per unit of energy that an EU $n$ asks from the SFC depends not only on the available energy $e_n$ to the EU but also on the the EU's sensitivity $\alpha_n$ to the choice of price, which is chosen motivated by the use of preference/reluctance parameters in \cite{Tushar-TIE:2014}, \cite{Wayes-J-TSG:2012}, and \cite{ZFan_Conf:2011,Samadi-C-Smartgridcomm:2010,chaibo-TSG:2014}. $\alpha_n$ captures the sensitivity of each EU's benefit from changing the per unit price and thus is used to quantify the different types of players. For example, if one user has a large amount of surplus, he may want to sell the energy at a relatively cheaper price compared to a player who has a small surplus to sell. This is due to the fact that players with larger amount of surplus energy might be more interested in selling all the energy for a higher gain and thus will be flexible in reducing their price for selling more energy. Hence, a relatively lower price may not affect their revenues significantly \cite{Tushar-TSG:2013}. On the other hand, a user with a small energy surplus will not be more keen to sell energy unless the price per unit of energy is considerably higher as otherwise the expected return will be very small. Hence, the evaluation of a change in price per unit of energy as well as the willingness to increase the price may not be same to both the players. We capture this aspect through a parameter $\alpha_n$, which is multiplied by the price per unit of energy (i.e., to capture the fact that a similar price may be interpreted differently by different EUs).

Also, as a rational entity each EU $n$ wants to increase the price $p_n$ per unit of energy that it charges the SFC as much as possible. However, the maximum price $P_n$ chosen by each $n$ needs to be  such that it does not exceed the grid's selling price $p_g$ per unit of energy. For example, if $P_n$ is greater than grid's selling price $p_g$, clearly the SFC will not buy any energy from the EU. To this end, each EU $n$ may want to increase $p_n$ to a maximum value of $P_n$ per unit of energy for selling its surplus and the choice of price $p_n$ is determined by 1) the surplus $e_n$ available to EU $n$, 2) EU's sensitivity $\alpha_n$ to the choice of price $p_n$, and 3) finally, the budget available to the SFC such that \eqref{eqn:budget-price} is satisfied.

Now, to determine the energy trading parameter $\mathbf{p} = [p_1, p_2,\hdots, p_n, \hdots, p_N]$, each EU $n$ interacts with the SFC to decide on the price $p_n$ per unit of energy that it wants to charge the SFC for selling its energy $e_n$ with a view to maximizing its benefit. To capture the benefit to each EU $n$, we propose a benefit function\footnote{Also known as utility function and welfare function.} $X_n$. In standard game theoretic research, e.g., \cite{Book_Han:2014}, the benefit function is an input to the game, whose outcome needs to be a real number, and illustrates the change in benefits corresponding to the change of a player's choice of action or environmental parameters. Note that benefit functions can be a combination of parameters with different dimensions and units to capture the effects of the change of parameters by the players. For instance, in \cite{Basar:2002}, the authors consider a utility function, which is a combination of transmission rate and the cost of transmission in order to show how different choices of price and transmission rate can affect the benefits to the respective player. In \cite{chaibo-TSG:2014}, the authors use a welfare function for their game, in which the welfare function is a combination of total cost of energy trading and the square of the amount of energy traded by the player. Further, a non-cooperative game is proposed in \cite{Yunpeng:2014}, in which a utility function is combining the cost of energy and the quantity of energy to be sold.

Now, we propose a benefit function $X_n$, which is based on a linearly decreasing marginal benefit,
\begin{eqnarray}
\tilde{X}_n = P_n - \alpha_n p_n - e_n,\label{eqn:marginal-utility}
\end{eqnarray}
contingent on the following assumptions:
\begin{enumerate}
\item Each EU $n$ is a rational entity and wants to choose a price $p_n$ per unit of energy as close to the maximum possible price $P_n$, e.g., equal to the grid's selling price, as possible.
\item EU $n$ is sensitive to its choice of price per unit of energy through the parameter $\alpha_n$, and thus the choice of $p_n$ is restricted by the choice of $\alpha_n$.
\item An EU with large surplus of energy has relatively lower marginal benefit compared to the EUs with lower surplus for the same choice of price~\cite{Tushar-globecom:2014}.
\end{enumerate}
Note that a utility function with decreasing marginal benefit is shown in \cite{Marginal_1:1998} to be appropriate for energy users. Then, the same property is also used to design utility functions in \cite{Tushar-TSG:2013}, \cite{Wayes-J-TSG:2012}, \cite{Samadi-C-Smartgridcomm:2010}, and \cite{chaibo-TSG:2014}, where the players participate in games for making decisions on energy trading parameters including price and energy. To this end, the chosen utility function in \eqref{eqn:utility-function} also possesses the property of declining marginal benefit and is close to the utility functions used in \cite{Tushar-TSG:2013}, \cite{Wayes-J-TSG:2012}, and \cite{chaibo-TSG:2014}. Please note that the authors modeled a two-level game in \cite{chaibo-TSG:2014} and Stackelberg games in \cite{Tushar-TSG:2013} and \cite{Wayes-J-TSG:2012} for designing energy management for smart grid. Therefore, a similar property is used to model the utility function of the proposed game.

To that end, the net benefit $X_n$ that an EU $n$ can attain from selling its surplus $e_n$ to the SFC at a price $p_n$ can be defined as
\begin{eqnarray}
X_n(p_n) &=& \int_0^{p_n}\!\tilde{X}_n\mathrm{d}q_n\nonumber\\&=&P_np_n - \frac{\alpha_n}{2}{p_n}^2-e_np_n.\label{eqn:utility-function}
\end{eqnarray}
The benefit function in \eqref{eqn:utility-function} is a quadratic function of $p_n$, which leverages the generation of discriminate pricing~\cite{Tushar-globecom:2014} between different EUs of the system. As can be seen in \eqref{eqn:utility-function}, the proposed benefit function possesses the following properties:
\begin{enumerate}
\item The benefit function is a concave function of $p_n$, i.e., $\frac{\delta^2X_n}{{\delta p_n}^2}<0$. Hence, the benefit of an EU may decrease for an excessively high $p_n$. This models the fact that the if the price is very high, the SFC may restrain from buying energy from the EU, which will eventually reduce its expected benefit from energy trading.
\item The benefit function $X_n$ is an increasing function of $P_n$ and a decreasing function of $\alpha_n$. That is $\frac{\delta X_n}{{\delta P_n}}>0$ and $\frac{\delta X_n}{{\delta \alpha_n}}<0$. Therefore, a EU with higher sensitivity will be prone not to change its selling price $p_n$ significantly.
\end{enumerate}
A graphical representation of the properties of the benefit function $X_n$ is shown in Fig.~\ref{fig:UtilityProperty}.
\begin{figure}[t]
\centering
\includegraphics[width=\columnwidth]{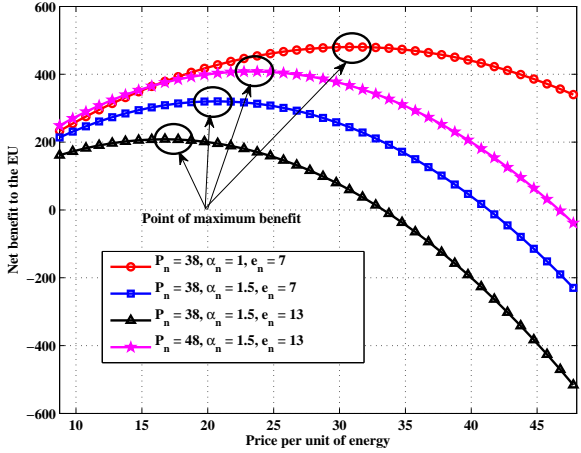}
\caption{Demonstration of the effects of different parameters on the net benefit achieved by an EU by trading its energy with the SFC. As can be seen from the figure, a higher sensitivity $\alpha_n$ leads the EU $n$ to achieve its maximum net benefit at a lower $p_n$. Also, as the surplus of an EU decreases, it achieves its maximum net benefit at a relatively higher price compare to other EUs.} \label{fig:UtilityProperty}
\end{figure}

It is important to note that for sustainable energy trading in a smart grid system the overall system benefits from trading needs to be profitable. Otherwise, EUs with significantly lower revenue from energy trading would not participate in such trading~\cite{Tushar-TSG:2013}, which will eventually diminish the energy trading market. In this context, the objective of each EU $n$ in the system is to interact with the SFC in order to determine a price $p_n$ per unit of energy, within the budget $C$ in \eqref{eqn:budget-price} of the SFC, so as to sell its energy surplus $e_n$ such that the sum $\sum_n X_n$ over all EUs in the system is maximized. Mathematically, the objective can be expressed as
\begin{eqnarray}
\max_{p_n} X(\mathbf{p}) = \max_{p_n}\left[\sum_n\left(P_np_n - \frac{\alpha_n}{2}{p_n}^2-e_np_n\right)\right], \label{eqn:objective-EU}
\end{eqnarray}
such that $\sum_n e_np_n\leq C$. Now, to explore how each EU $n$ in the proposed system can identify a price $p_n$ within the budget $C$ in \eqref{eqn:budget-price} such that their objective in \eqref{eqn:objective-EU} can be attained, we propose a CCG in the next section.

We stress that the proposed problem can also be solved by using other techniques such as dual decomposition or other centralized schemes. In order to solve the problem in a distributed fashion and thus allow the EUs  to maintain their privacy (i.e., not to disclose private information like $\alpha_n$ and $P_n$ to  the SFC and other EUs in the network), we choose to use a CCG over other centralized techniques. As discussed in \cite{privacy_2012}, sharing such private information not only allows the SFC to control the EUs' energy usage behavior, but also enables the SFC to access  the EUs' private lifestyles, which is a significant privacy concern at present. Hence, certain information, such as $\alpha_n$ (the sensitivity of a user to the change of price), needs to be kept private and should not be shared with the SFC. In addition, although the SFC may solve the proposed problem in a centralized manner if it has access to the private information, the EUs may be concerned that the SFC could modify $\alpha_n$, and thus distribute the total budget in favor of the EUs that have good relationships with the SFC, which is especially possible if the EUs and the SFC do not trust each other. Therefore, considering the above-mentioned factors, we choose to use a distributed technique such as the proposed CCG that leads to a solution with desirable properties like social optimality and being envy-free, while at the same time preserving the privacy of the EUs.

\section{Cake Cutting Game}\label{sec:game}
\subsection{Brief Background}\label{sec:Background}
CCG is a branch of game theory that deals with the division of some finite pool of resources in a way that meets certain valuation criteria or objectives of the players splitting the resource~\cite{Batral-2014}. Formally, the cake can be represented as a convex set, which is the total budget $C\subset\mathbb{R}$ of the SFC in this proposed case. Each player $n$ will receive an allocation $p_ne_n$ of this $C$ by choosing a suitable $p_n$ with the property $p_ne_n\subseteq C$. Now, before proceeding to the design of the proposed game, first we discuss some key properties of a CCG~\cite{Batral-2014}, which are relevant to the proposed scheme as follows.
\begin{enumerate}
\item The allocation vector $\mathbf{p^*} = [p_1^*, p_2^*, \hdots, p_N^*]$, which contains the outcome received by each player $n\in\mathcal{N}$ via distributing the cake $C$ is \emph{complete} if
\begin{eqnarray}
\sum_n e_n p_n^* = C.\label{eqn:prop-complete}
\end{eqnarray}
\item An allocation is called a \emph{socially optimal} allocation if the allocation has the property
\begin{eqnarray}
X (\mathbf{p}^*)\geq X(\mathbf{p}),\label{eqn:prop-social-optimal}
\end{eqnarray}
where $\mathbf{p} = [p_1^*, p_2^*,\hdots, p_n,\hdots,p_N^*]$ for any $n\in\mathcal{N}$.
\item An allocation $\mathbf{p}^*$ possesses the property of \emph{Pareto optimality} if no player with an allocation $p_{n1}^*$ can be better off with a share $p_{n1}$ without hurting at least one other player. Mathematically, an allocation with $p_{n1}^*$ and $p_{n2}^*$ is Pareto optimal, if there exists no other allocation containing $p_{n1}$ and $p_{n2}$ such that
\begin{eqnarray}
X_{n1}(p_{n1}^*)\leq X_{n1}(p_{n1})\wedge X_{n2}(p_{n2}^*)\leq X_{n2}(p_{n2}),\nonumber\\~\forall n1 = 1,2,\hdots,N,~\exists~n2 = 1, 2, \hdots, N; p_{n1}^*,p_{n2}^*\in\mathbf{p}^*. \label{eqn:prop-pareto}
\end{eqnarray}
\item If a \emph{complete allocation} $\mathbf{p}^*$ of a CCG is socially optimal, it is also Pareto optimal~\cite{Batral-2014}.\label{property:4}
\end{enumerate}
\subsection{Proposed Game}\label{sec:Proposed Game}
To decide on the energy trading parameter $p_n$, each EU in the smart community interacts with the SFC, and we propose a CCG $\Pi$ to capture this interaction. In the proposed CCG $\Pi$, each EU $n$ decides on its selling price $p_n$ through the considered game and offers the price to the SFC. The SFC, on the other hand, compares the received price vector $\mathbf{p}$ from all the EUs and decides whether the total expense is within its budget $C$, i.e., if the expense satisfies \eqref{eqn:budget-price}, and thus decides the solution of the game. Formally. the proposed CCG $\Pi$ can be defined by its strategic form as\footnote{Since the EUs do not have storage facilities, all EUs with energy surplus will be interested in participating in the game to make revenue as long as the offered price is more than the grid's buying price. Hence, the proposed CCG falls within the example of game theoretic problems where competitive buyers participate in the game to increase their utilities, e.g., studies in \cite{Wayes-J-TSG:2012} and \cite{Samadi-C-Smartgridcomm:2010}, \cite{chaibo-TSG:2014}.}.
\begin{eqnarray}
\Pi = \{\mathcal{N},\{\mathbf{P}_n\}_{n\in\mathcal{N}}, X\},\label{eqn:CCG}
\end{eqnarray}
where $\mathcal{N}$ is the set of players, i.e., EUs, in the game and $\mathbf{P}_n$ is the set of strategies of EU $n$, i.e., $p_n\in\mathbf{P}_n,~\forall n$, that satisfy \eqref{eqn:budget-price}. In \eqref{eqn:CCG}, the choice of each EU $p_n\in\mathbf{P}_n$ affects the choice of other EUs in choosing their suitable selling price due to the presence of \eqref{eqn:budget-price}. Hence, the proposed CCG $\Pi$ can be considered as a \emph{variational inequality} problem $\Gamma$~\cite{VIP:2007}, in which the choice of strategies of multiple players are coupled through \eqref{eqn:budget-price}. Hence, variational equilibrium, which is the solution concept of variational inequality problems, can be presumed as the solution of the proposed CCG $\Pi$. For the rest of the paper, we will use the term \textbf{cake cutting solution (CCS)} to refer to a variational equilibrium of the proposed CCG $\Pi$.
\begin{definition}
Consider the CCG $\Pi$ in \eqref{eqn:CCG}, in which $X$ is defined by \eqref{eqn:objective-EU}. A set of strategies $\mathbf{p}^*$ constitutes the CCS of the proposed $\Pi$ if and only if the strategy set satisfies the following inequalities:
\begin{eqnarray}
X(p_n^*,\mathbf{p}_{-n}^*)\geq X(p_n,\mathbf{p}_{-n}^*),\nonumber\\\forall n\in\mathcal{N}, p_n\in\mathbf{P}_n~\forall n,~\sum_n e_n p_n\leq C,\label{eqn:CCS}
\end{eqnarray}
where $\mathbf{p}_{-n} = \left[p_1, p_n,\hdots, p_{n-1}, p_{n+1}, \hdots, p_N\right]$.
\end{definition}
\subsection{Properties of CCS}\label{sec:Properties of The Game}
In this section, we investigate the properties of the CCS. In particular, we determine whether there exists a socially optimal CCS of the proposed CCG $\Pi$. Essentially, a socially optimal solution maximizes the total benefit to all the EUs in the smart grid, and thus is suitable for allocation of $C$ in order to maximize the overall system benefit.
\begin{theorem}
There exists a socially optimal CCS of the proposed CCG $\Pi$ between the EUs in the smart grid.
\label{theorem:1}
\end{theorem}
\begin{proof}
First we note that the proposed CCG $\Pi$ is a variational inequality problem $\Gamma$ as we have mentioned earlier. Therefore, the CCG $\Pi$ can be defined as $\Gamma(\mathbf{P, Z})$, which can be used to determine a vector $\mathbf{p}\in\mathbf{P}\in\mathbb{R}^n$ such that $\langle\mathbf{Z}(\mathbf{p}^*),\mathbf{p}-\mathbf{p}^*\rangle\geq 0,~\forall \mathbf{p}\in\mathbf{P}$~\cite{Agranda:2011}. Here,
\begin{eqnarray}
\mathbf{Z} = -\left(\nabla_{p_n}X_n(p_n)\right)_{n\in\mathcal{N}},\label{eqn:defin-Z}
\end{eqnarray}
and $\mathbf{P}$ is the vector of all strategies of all EUs in the network. The solution of the $\Gamma(\mathbf{P, Z})$ is the CCS. Now, the pseudo-gradient of the benefit function \eqref{eqn:objective-EU} is~\cite{Agranda:2011}
\begin{eqnarray}
\mathbf{Z} = \begin{bmatrix}
\alpha_1p_1 + e_1 - P_1\\
\alpha_2p_2 + e_2 - P_2\\
\vdots\\
\alpha_Np_N + e_N - P_N
\end{bmatrix},
\label{eqn:definition-Z}
\end{eqnarray}
whose Jacobean is
\begin{eqnarray}
\mathbf{JZ} = \begin{bmatrix}
\alpha_1 & 0 & \cdots & 0\\
0 & \alpha_2 & \cdots & 0\\
\vdots & \vdots & \cdots & \vdots\\
0 & 0 & \cdots & \alpha_N
\end{bmatrix}.
\label{eqn:Jecobean-Z}
\end{eqnarray}
In \eqref{eqn:Jecobean-Z}, $\alpha_n$ for $n = 1, 2, \hdots, N$ is always positive. Now, by considering the $n^\text{th}$ leading principal minor $\mathbf{JZ}_n$ of the leading principal sub-matrix\footnote{The $i^\text{th}$ order principal sub-matrix $\mathbf{A}_i$ can be obtained by deleting the last $g-i$ rows and last $g-i$ columns from a $g\times g$ matrix $\mathbf{A}$.} $\mathbf{JZ}$, it can be shown that $\mathbf{JZ}_n$ is always positive, i.e., $|\mathbf{JZ}_1|>0,~|\mathbf{JZ}_2|>0$ and so on. Therefore, $\mathbf{JZ}$ is positive definite on $\mathbf{P}$, and thus $\mathbf{Z}$ is strictly monotone. Hence, $\Gamma(\mathbf{P, Z})$ possesses a unique CCS~\cite{VIP:2007}. Furthermore, due to the presence of the joint constraint \eqref{eqn:budget-price}, the CCS is also the unique global maximizer of \eqref{eqn:objective-EU}~\cite{VIP:2007}, which subsequently proves that the CCS is the \emph{socially optimal} solution of the proposed CCG $\Pi$.
\end{proof}

Whereas the CCG $\Pi$ is shown to have a socially optimal solution, in order to divide the budget $C$ among the EUs in an efficient manner it is also necessary that the budget should be Pareto optimal. We note that in a Pareto efficient allocation no EU can change its strategy without hurting at least one another EU in the network. Therefore, if the allocation is both socially and Pareto optimal, the allocation of price per unit of energy between different EUs will be fair and envy-free. To this end, first we note that the social optimality of the CCS has already been shown in Theorem~\ref{theorem:1}. Therefore, according to Property~\ref{property:4} in Section~\ref{sec:Background}, demonstrating the existence of a complete allocation of price between the EUs will subsequently establish the Pareto optimality of the CCS.

Now, due to the coupled constraint \eqref{eqn:budget-price}, the Karush-Kuhn-Tucker (KKT) condition, using the method of Lagrange multipliers~\cite{Bertsekas:1995}, for the EU $n$'s choice on $p_n$ in \eqref{eqn:objective-EU} can be defined as~\cite{VIP:2007}
\begin{eqnarray}
P_n - \alpha_n p_n - e_n -\tau_n = 0, \tau_n\geq 0,\label{eqn:lagran-1}
\end{eqnarray}
and
\begin{eqnarray}
\tau_n\left(\sum_n e_n p_n-C\right) = 0,~\forall n.\label{eqn:lagran-2}
\end{eqnarray}
Here, $\tau_n$ is the Lagrange multiplier for EU $n$. It is important to note that if any strictly monotone variational inequality problem constitutes a CCS such as the proposed case (according to Theorem~\ref{theorem:1}), the multiplier $\tau_n,~\forall n$ possesses the property $\tau_n>0,~\forall n$~\cite{VIP:2007}. As a consequence, it is clear from Theorem~\ref{theorem:1} and \eqref{eqn:lagran-2} that for the proposed CCG $\Pi$ the total allocation of budget between the EUs is equal to $C$, i.e., $\sum_n p_ne_n = C$. Thus, the allocation is \emph{complete}, and consequently the CCS of the proposed CCG $\Pi$ is Pareto optimal.
\begin{remark}
Since, the solution of the CCG $\Pi$ possesses a solution, which is both socially and Pareto optimal, the allocation of price per unit of energy between different EUs will be fair and envy-free\footnote{Such an envy-free property allows each EU to trade its surplus energy with the SFC without envying other EUs in the network for the payments they receive, and thus ensures market transparency despite price discrimination and consequently enables such an energy trading market to be sustained.}.
\end{remark}

Another important characteristic of the decision making process of a EU $n$ concerning its choice of a price $p_n$ per unit of energy can be explained from \eqref{eqn:lagran-1}. Since, $\tau_n> 0~\forall n$,
\begin{eqnarray}
p_n<\frac{P_n - e_n}{\alpha_n}. \label{eqn:chosen-price}
\end{eqnarray}
Thus, an EU with higher energy surplus and/or higher sensitivity to the choice of price needs to choose a relatively smaller price per unit of energy compared to EUs with lower surplus and/or lower sensitivity in order to reach the socially and Pareto optimal CCS if the proposed CCG $\Pi$ is adopted.
\subsection{Algorithm}\label{sec:Algorithm}
Now, to design an algorithm for the EUs to reach the desired solution of the CCG $\Pi$, we first note that the proposed game is a strictly monotone variational inequality problem. Therefore, the CCS can be attained by solving the game through an algorithm, which is suitable for solving a monotone variational inequality problem. To this end, we propose to use the S-S method~\cite{S-S:1999}, which is shown to be effective to solve monotone variational inequality problems in \cite{Tushar-TSG:2013}, to solve the proposed CCG $\Pi$. Essentially, the S-S method used in this paper is a hyperplane projection method that requires two-way communications between EUs and the SFC in the network to reach the CCS. A geometrical interpretation is used where two projections per iteration are required. For instance, let $p^t$ be the current approximation of the solution of $\Gamma(\mathbf{P, Z})$. Then, first the projection $r(p^t) = \text{Proj}_{\mathbf{P}}[p^t - \mathbf{Z}(p^t)]$ is computed\footnote{Projection $\text{Proj}_{\mathbf{P}}[k] = \arg\min\{||w-k||,~w\in\mathbf{P}\},~\forall k\in\mathbb{R}^n$~ \cite{Tushar-TSG:2013}.} and a point $z^t$ is searched for in the line segment between $p^t$ and $r(p^t)$ such that the hyperplane $\partial H  \triangleq \{p\in\mathbb{R}|\langle\mathbf{Z}(k^t),p-z^t\rangle=0\}$ strictly separates $p^t$ from the solution $\mathbf{p}^*$ of the problem. Then, once $\partial H$ is constructed, $p^{t+1}$  is computed in iteration $t+1$ by projecting onto the intersection of the feasible set $\mathbf{P}$ with hyperspace $H^t\triangleq\{p\in\mathbb{R}|\langle\mathbf{Z}(k^t),p-z^t\rangle\leq0\}$, which contains the solution set.
\begin{algorithm*}[t]
\caption{Algorithm to reach the CCS of the proposed CCG $\Pi$.}
\centering
\includegraphics[width=0.8\textwidth]{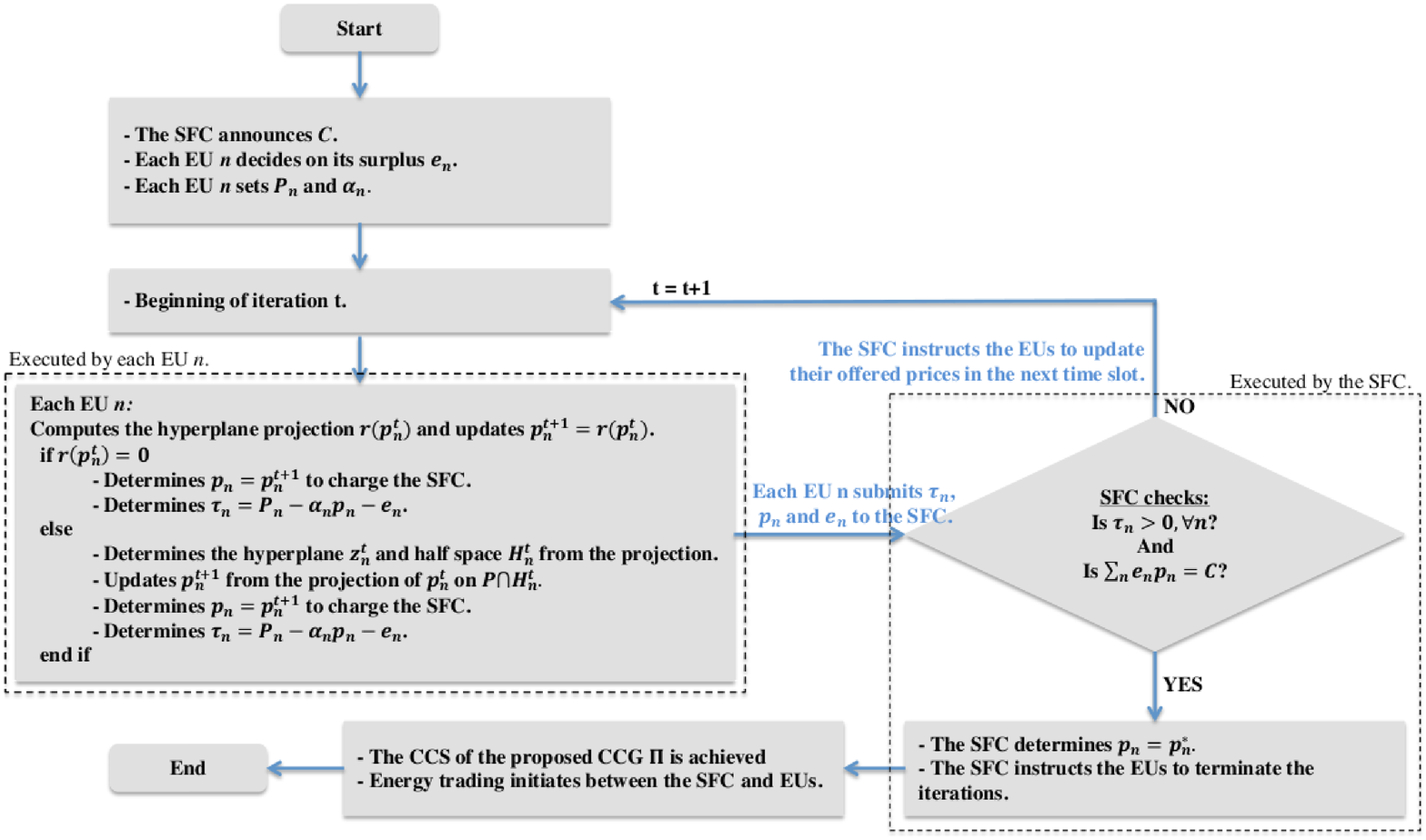}
\label{algorithm:1}
\end{algorithm*}

The proposed algorithm is initiated with the announcement of a total budget by the SFC to buy electricity from the connected EUs. The SFC can set the  budget using any statistical technique such as a Markov chain model based on historical budget data sets~\cite{Shamshad-Energy:2005}. Upon receiving the information about the budget and determining its own requirement, each EU decides the amount of energy that it wants to sell and submits it to the SFC. Once, the budget is set and the surplus of each EU is determined, all the EUs participate in the proposed CCG $\Pi$ through Algorithm~\ref{algorithm:1} and reach the optimal price vector, i.e., the CCS $\mathbf{p}^*$, for selling their surplus energy to the SFC. The details of the algorithm in shown in \ref{algorithm:1}.
\begin{proposition}
The algorithm proposed in Algorithm~\ref{algorithm:1} is always guaranteed to reach the CCS of the proposed CCG $\Pi$.
\label{proposition:1}
\end{proposition}
\begin{proof}
To prove Proposition~\ref{proposition:1}, first we note that the S-S method is based on a hyperplane projection technique~\cite{Agranda:2011}, which is always guaranteed to converge to a non-empty solution if the variational inequality problem is strictly monotone~\cite{VIP:2007}. It is proven in Theorem~\ref{theorem:1} that the proposed CCG $\Pi$ is strictly monotone over the choice of $p_n~\forall n$. Therefore, the proposed Algorithm~\ref{algorithm:1} is guaranteed to always reach a non-empty CCS, and thus Proposition~\ref{proposition:1} is proved.
\end{proof}

\textbf{Note}: Please note that the proposed game is a static game and therefore does not consider the effect of the change of parameters across different time slots. Here, we consider a single time instant and keep the entire focus of the study on investigating how the considered energy trading scheme can be conducted in an envy-free environment so as to achieve a socially optimal solution by using the proposed price discrimination technique. Once such price discrimination is established, extending the proposed work to a time-varying environment is an interesting topic for future work. Note that the electricity price in real-time pricing schemes is decided differently at different times of the day based on the conditions of several parameters such as the demand, electricity generation, and the reserve of energy in the system. Therefore, the proposed scheme has the potential to be incorporated into such a real-time pricing scheme, in which the energy controller may decide to adopt the price discrimination at a particular time of the day, whenever it seems beneficial.

\section{Case Study}\label{sec:Case Study}
To show the effectiveness of the proposed scheme, we consider an example in which a number of EUs with energy surplus are participating in supplying energy to the SFC in a time of interest. The energy surplus of each EU to supply to the SFC is assumed to be a uniformly distributed random variable in the range $[3.6, 12.25]$ kWh~\cite{Tushar-TSG:2013}. The target per unit price $P_n$ for all $n\in\mathcal{N}$ is assumed to be $45$ cents, which is in fact equal to the grid's selling price~\cite{Tushar-TIE:2014}. This value is based on the rationality assumption of the EUs where each EU is willing to charge the SFC as much as the grid for selling its energy. The budget of the SFC is considered to be $1000$ cents, which is chosen to maintain the condition $p_ne_n\leq C<P_ne_n$ throughout the simulation process. Further, the choice of $C$ is also chosen such that the price $p_n$ for any EU $n$ does not go below the grid's buying price $p_{g,\text{buy}} = 8.00$ cents/kWh~\cite{Tushar-TIE:2014}. This is necessary to ensure that all EUs are interested in selling to the SFC instead of to the grid. The sensitivity parameter $\alpha_n~\forall n$ is chosen randomly from the range $[1,~3]$. Thus, a consumer with $\alpha_n = 1$ is least sensitive to its choice of price whereas consumers with sensitivity $\approx 3$ are considered to be strictly constrained to the price choice. Nevertheless, it is important to note that all parameters used in this study are particular for this example only and may vary according to the needs of the SFC, weather conditions, time of day/year and the energy policy of the particular country.
\begin{figure}[t]
\centering
\includegraphics[width=\columnwidth]{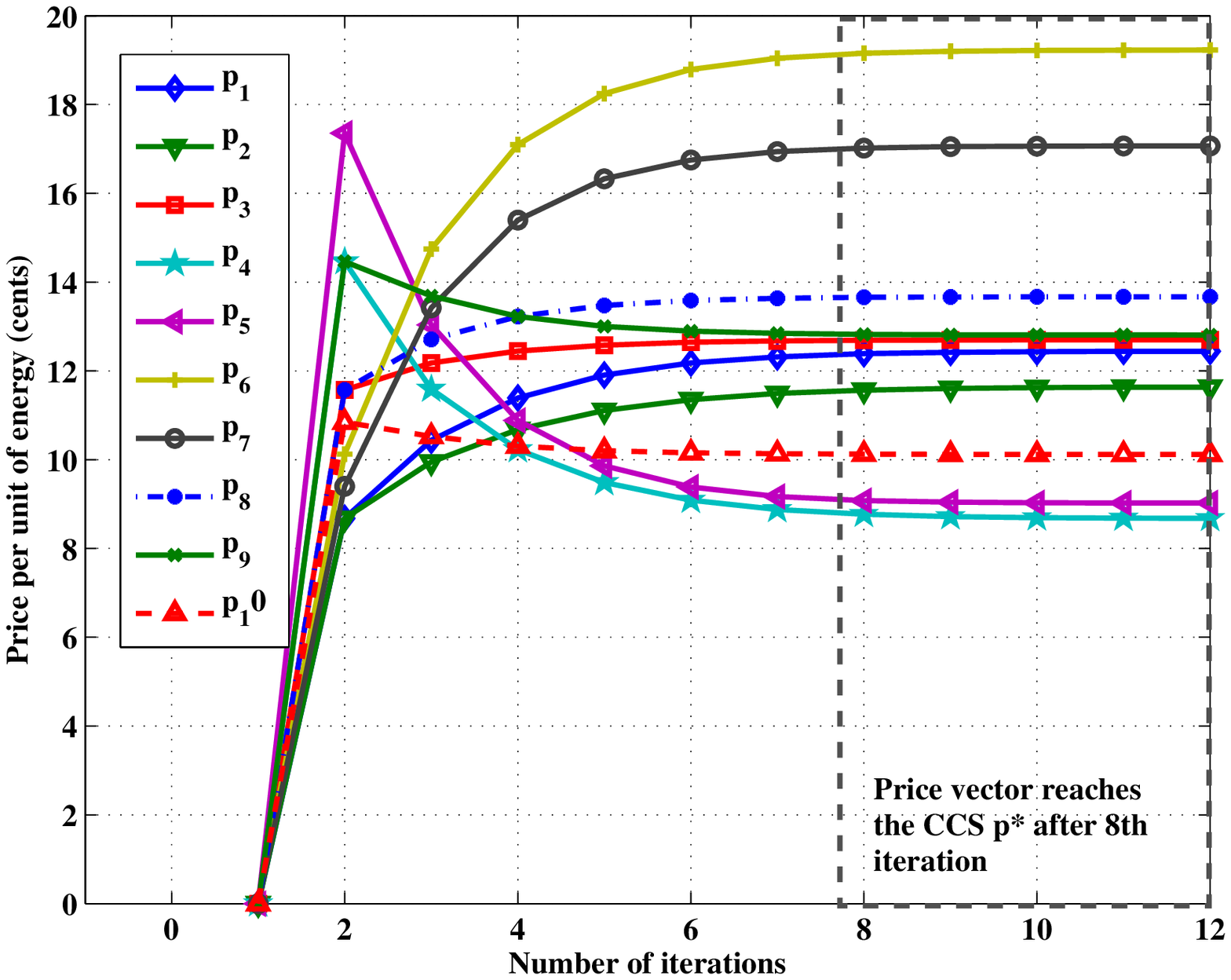}
\caption{Demonstration of the convergence of the proposed Algorithm~\ref{algorithm:1} to the CCS of the proposed CCG $\Pi$. It is noted from the figure that the price vector reaches its CCS after $8^\text{th}$ iterations.} \label{fig:figure-1}
\end{figure}

To that end, in Fig.~\ref{fig:figure-1}, we first show  how the proposed game with $10$ EUs converges to its CCS solution when we adopt Algorithm~\ref{algorithm:1}. Firstly, according to this figure, the choice of price of each of the participating EUs reaches its CCS after the $8^\text{th}$ iteration of the algorithm. Hence, the speed of the algorithm is reasonable. Secondly, different EUs determine different prices to pay to the SFC for energy trading, which is mainly due to the way that the pricing scheme is designed. Note that although the target price $P_n~\forall n$ per unit of energy is considered similar for all EUs, the energy surpluses available to them are different. Also, different EUs have differing sensitivity to the pricing policy. As a consequence, once the CCG $\Pi$ reaches the CCS, the socially optimal price vector constitutes a different price per of unit energy for each EU as demonstrated in Fig.~\ref{fig:figure-1}.
\begin{figure}[t]
\centering
\includegraphics[width=\columnwidth]{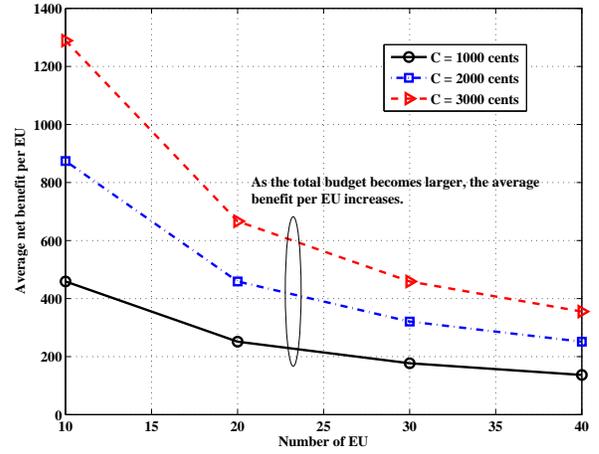}
\caption{Demonstration of the effect of the total budget of the SFC and the number of EUs that are sharing the budget on the average net benefit achieved by each EU participating in the proposed CCG $\Pi$.} \label{fig:figure-2}
\end{figure}

It is important to note that the outcome of the proposed CCG $\Pi$ is significantly affected by the size of the cake of the game, i.e., the total budget available to the SFC, and the number of participating EUs that are sharing the cake between them through choosing a suitable price per unit of energy. Now, to show the effect of the number of EUs that are sharing the budget of the SFC through the proposed CCG $\Pi$, we consider an example where four different of EUs are sharing the SFC's budget in order to maximize their net benefits by choosing a suitable price per unit of energy. The graphical representation of this considered study is shown in Fig.~\ref{fig:figure-2}. We change the number of EUs from $10$ to $40$ with an increment of $10$, and show the average net benefits achieved by each EU of each respective case in the figure. According to Fig.~\ref{fig:figure-2}, first we note that as the number of EUs increases, the average benefit achieved by each EU decreases due to sharing the same budget of the SFC. For instance, as the number of EUs increases from $10$ to $40$, the average benefit for sharing $1000$ cents by choosing a  suitable price per unit of energy decreases from $1250$ to $390$, which is around a  $68\%$ decrement. Essentially, more EUs taking part in energy management enables each EU to take a smaller share of the budget, which subsequently reduces their benefits in energy trading with the same total budget.
\begin{table}[t]
\caption{Demonstration of the effect of SFC's total budget on the participation rate of  EUs (with surplus energy) in energy trading. It is considered that if any EU does not receive a minimum payment of $p_{g,\text{buy}} = 8.00$ cents/kWh it does not trade its energy with the SFC. Thus, a lower budget for a large number of EUs subsequently decreases the percentage rate of the EUs' participation in energy trading with the SFC. Here, the actual numbers of EUs that participate in energy trading are shown whereby their percentages compared to the total EUs in the network are shown within the brackets.}
\centering
\includegraphics[width=\columnwidth]{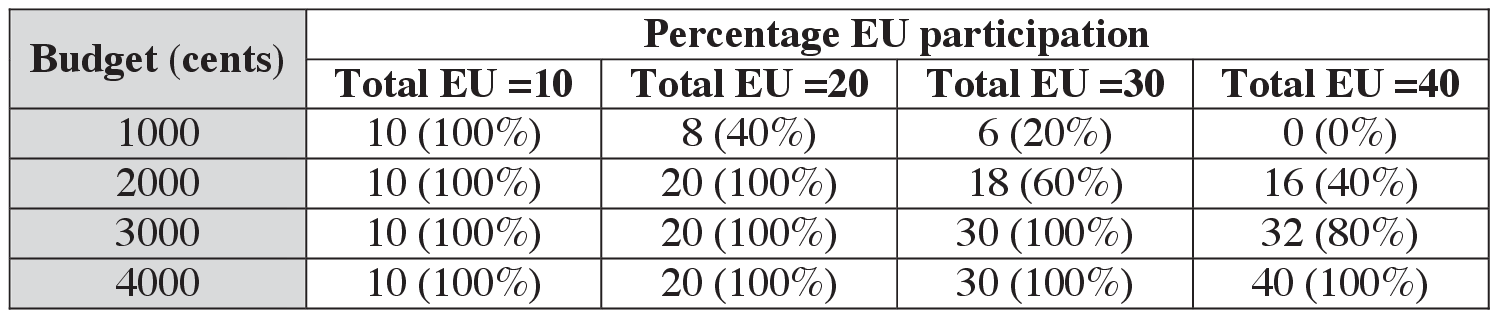}
\label{fig:table-budget}
\end{table}

Another interesting aspect that can be noticed from the effect of the SFC's budget on the overall energy trading scheme is that the same budget may not be suitable to encourage all the EUs from EU groups of different sizes to trade their energy with the SFC. In this regard, we show the participation rate of EUs from different sizes of EU groups in Table~\ref{fig:table-budget}. It is assumed that if the price per unit of energy that an EU charges the SFC for selling its energy falls below $p_{g,\text{buy}}$, which is considered to be $8.00$ cents/kWh for this particular case study, it does not participate in energy trading. This is due to the fact that, as explained in Section~\ref{sec:system-model}, the expected return from selling energy for the EU becomes very small. Now, according to Table~\ref{fig:table-budget}, for a similar budget $C$ of the SFC, the participation rate of EUs reduces considerably as the number of EUs in a group increases. For instance, for a budget of $1000$ cents, $100\%$ EU participation is observed from a group of $10$ EUs, whereby the participation rate decreases to $40\%$, $20\%$ and $0\%$ respectively for EU group size of $20, 30$ and $40$. Nonetheless, by increasing the budget $C$, the SFC can encourage more EUs to be involved in the energy trading with the SFC and can reduce its total cost of energy trading. Thus, as shown in Table~\ref{fig:table-budget}, the proposed scheme can essentially assist the SFC in deciding on its budget with a view to increase participation, if appropriate, with the SFC.
\begin{remark}
It is important to note from the above discussion in Fig.~\ref{fig:figure-2} and Table~\ref{fig:table-budget} that a suitable choice of the SFC's budget $C$  is critical to successful adaptation of energy management in the community through the proposed CCG $\Pi$ as this may possibly affect the average net benefit per EU as well as the total cost incurred by the SFC. The proposed scheme has the potential to assist the SFC in deciding on its budget in order to encourage more EUs, if feasible, to take part in energy trading with the SFC.
\end{remark}
\begin{table}[t]
\caption{Demonstration of the effect of the sensitivity parameter of EUs on their average net benefit. For this particular example, all $10$ EUs are considered to possess the same sensitivity parameter, i.e., all EUs have the same sensitivity to their chosen price per unit of energy. The average net benefit per EU is compared, i.e., percentage decrement as shown within brackets, with the case when all EUs possess $\alpha_n = 1,~\forall n$.}
\centering
\includegraphics[width=\columnwidth]{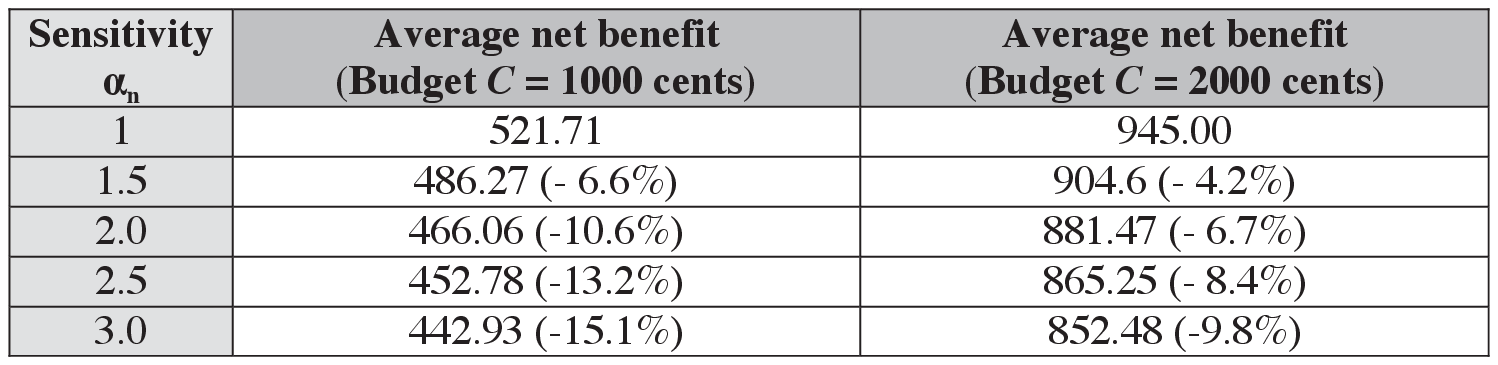}
\label{fig:table-1}
\end{table}

Furthermore, the choice of price and consequently the obtained net benefit by each EU is also affected by its sensitivity parameter $\alpha_n~\forall n$, which we show in Table~\ref{fig:table-1}. For this case, we assume that all EUs in the system are equally sensitive to their chosen price per unit of energy, i.e., they have the same $\alpha_n,~\forall n$, whereby their available surpluses to sell to the SFC are different as in previous examples. We consider five sensitivity parameters including $1.0, 1.5, 2.0, 2.5$ and $3.0$, where $1.0$ refers to the case when the EUs are insensitive to the choice of price per unit of energy and $3$ indicates maximum sensitivity of the EU. The average net benefit to the EU at $\alpha_n = 1$ is considered as the baseline and is used to compare how the average net benefit to the EUs varies as their sensitivity is altered in the system. As can be seen in Table~\ref{fig:table-1}, the average net benefit to the EUs decreases as the sensitivity increases. For a budget of $1000$ cents, for instance, as the sensitivity parameter increases from $1$ to $3$, the average net benefit to the EU reduces by $15.1\%$. In fact, as the sensitivity increases, the choice of price becomes more restricted for an EU, which consequently reduces its net benefit. A reduction in net benefit with increasing sensitivity is also observed for a budget of $2000$ cents. However, for the higher budget, we find two modifications in terms of achieved average net benefit per EU. Firstly, the benefit per EU increases as the budget of the SFC increases, which is explained in Fig.~\ref{fig:figure-2}. Secondly, the decrement of average net benefit per EU for different sensitivity parameters compared to $\alpha_n = 1$ is relatively lower for a larger budget. For example, as the sensitivity parameter increases from $1$ to $3$, the average net benefit per EU reduces to $9.8\%$ for a budget of $2000$ cents, whereas this reduction is $15.1\%$ for $1000$ cents. This is due to the fact that although $\alpha_n$ is considered the same for all EUs, the available surplus of each EU is different, which enables discriminate pricing. Now, based on the available surplus, each EU chooses its price per unit of energy, which increases for a larger budget (as shown in Fig.~\ref{fig:figure-2}). As a consequence, the difference between net benefits reduces at a higher budget compared to the case of a lower budget.

Having demonstrated some properties of the proposed discriminate pricing scheme, we now show how the proposed scheme can be beneficial to both the EUs and the SFC compared to the traditional case when both parties trade their energy with the grid. To show this comparison, we assume that the selling and buying prices per unit of energy set by the grid are $44$ and $8$ cents per kWh as discussed earlier in this section. To this end, we first consider in Table~\ref{fig:table-SFC-benefit} how much energy the SFC can buy within its budget if it only buys from the grid, compared to the case in which the SFC buys from the users through the proposed scheme. Then, in Table~\ref{fig:table-EU-benefit}, we show the total monetary benefit EUs can attain if they choose to sell their surplus to the SFC instead of selling the surplus to the grid.

\begin{table}[t]
\caption{Demonstration of the benefit to the SFC in terms of the amount of energy that the SFC can get from the proposed trading scheme compared to the case when the SFC only buys its energy from the grid for the same budget.}
\centering
\includegraphics[width=\columnwidth]{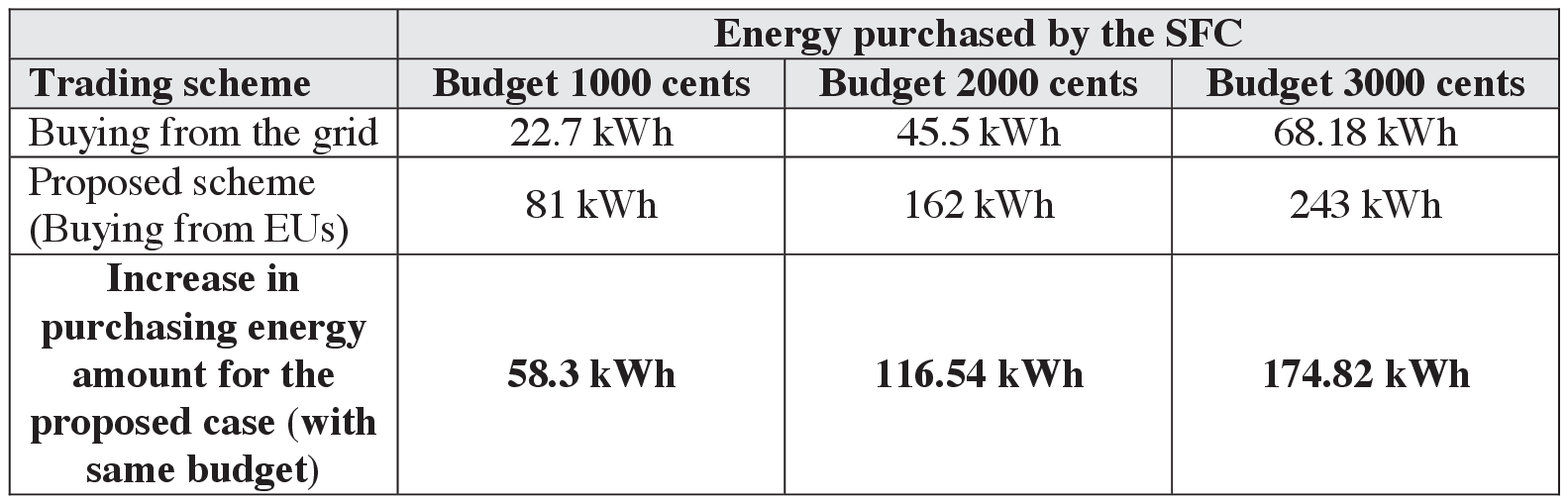}
\label{fig:table-SFC-benefit}
\end{table}

In Table \ref{fig:table-SFC-benefit}, we show the amount of energy that an SFC can buy from EUs if the proposed scheme is adopted as the SFC changes its budget from $1000$ cents to $2000$ and then $3000$ cents. We note that as the SFC increases its budget, more EUs would be interested in taking part in energy trading with the SFC (e.g., as shown in Table~\ref{fig:table-budget}), which subsequently increases the amount of energy that the SFC can obtain from the participating EUs. A similar increment in purchasing energy is also observed in the case when the SFC buys its energy from the main grid. However, for each budget, the energy that the SFC can buy from the EUs is considerably larger than the amount that the SFC can buy from the grid. For instance, for a budget of $1000$ cents, the SFC can buy $58.3$ kWh more energy through the proposed scheme compared to buying exclusively from the grid. This is due to the fact that the grid price is generally very high~\cite{McKenna-JIET:2013}. Hence, for a fixed budget $C$, the SFC can buy relatively smaller amounts of energy from the grid. As a consequence, the SFC manages to buy considerably more energy with the same budget. This phenomenon is observed for all considered SFC's budgets as shown in Table~\ref{fig:table-SFC-benefit}.
\begin{table}[t]
\caption{Demonstration of the benefits to the EUs in terms of total monetary revenue that the participating EUs receive when they trade their surplus energy with the SFC compared to trading energy with the grid.}
\centering
\includegraphics[width=\columnwidth]{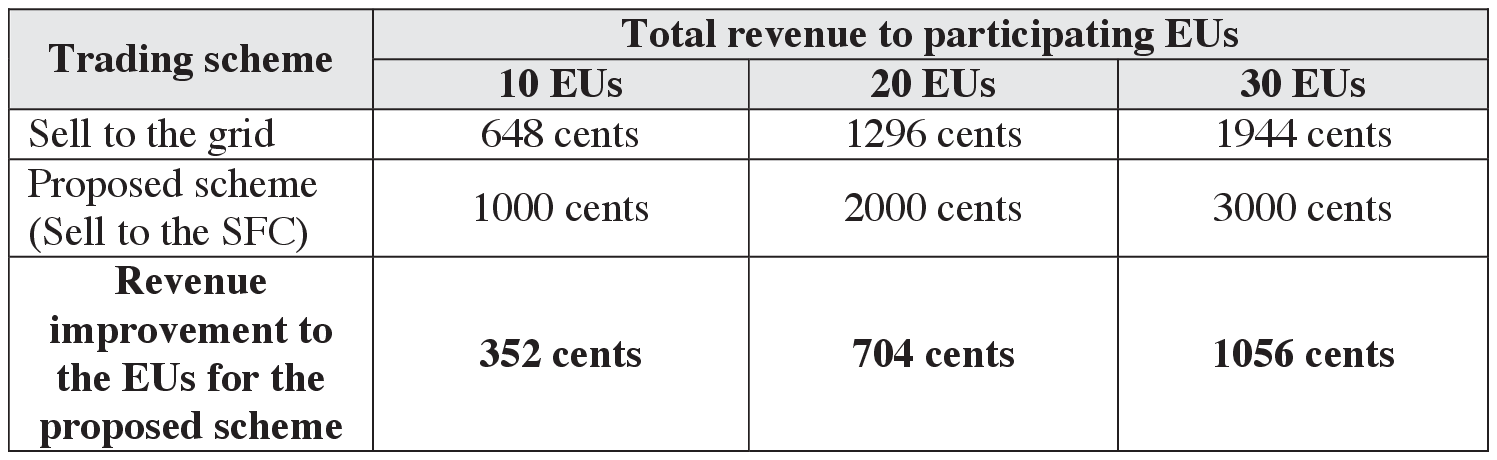}
\label{fig:table-EU-benefit}
\end{table}

In Table~\ref{fig:table-EU-benefit}, we show how the participating EUs in the proposed scheme can benefit in terms of total revenue that they can receive from trading their surplus energy to the SFC compared to selling to the grid. To this end, we first note that the proposed energy scheme through CCG $\Pi$ is \emph{complete} as discussed in Section~\ref{sec:Proposed Game}. Therefore, the total revenue that the participating EUs receive, for $10$, $20$ and $30$ EUs in the network,  becomes equal to the considered budgets of $1000$, $2000$, and $3000$ cents respectively. By contrast, as the EUs trade their energy with the grid the total revenue reduces significantly because of the lower buying price per unit of energy from the grid. For example, as $10$ EUs participating in the CCG $\Pi$, the total revenue that they receive is equal to the $1000$ cents budget of the SFC. Nonetheless, as the EUs trade their energy with the main grid, due to a lower per unit price of $8$ cents/kWh, the total revenue that the participating EUs receive reduces to $648$ cents, which is $352$ cents less than the proposed scheme for the considered system parameters. The performance improvement for the proposed case is better for higher numbers of EUs in the network.
\begin{remark}
It is clear from Table~\ref{fig:table-SFC-benefit} and Table~\ref{fig:table-EU-benefit} that the proposed scheme is beneficial for both the SFC and the participating EUs in the smart grid network for the system parameters considered in the given case studies. Hence, the proposed scheme has the potential to be adopted in practical systems in order to benefit all participating entities.
\end{remark}

\section{Conclusion}\label{sec:Conclusion}
This paper has demonstrated a viable method to discriminate price per unit of energy between different energy users in a smart grid system when the EUs sell their surplus energy to a shared facility controller. A cake cutting game has been proposed to leverage the generation of discriminate pricing within a constrained budget of the SFC. To study the fairness of the proposed scheme, it has been shown that the CCG can be modeled as a variational inequality problem that possesses the solution of the game, i.e., the cake cutting solution. The properties of the CCS have been studied and the existence of a socially optimal solution, which is also Pareto optimal, has been validated. An algorithm has been proposed that can be adopted by each EU interacting with the SFC in a distributed manner and the convergence of the algorithm to the optimal CCS has been confirmed. Finally, the properties of the game have been studied, and the advantages of discriminate pricing for both the SFC and the EUs have been demonstrated via simple comparisons with energy trading with the main grid.

An important extension of the proposed scheme would be to establish a relationship between the budget of the SFC and the total number of EUs participating in the energy trading with a view to enabling efficient price discrimination, e.g., by using an interactive Stackelberg game with imperfect information of the total budget. Another potential extension is to conduct studies that determine when such a discriminate pricing technique can be used as a real-time pricing scheme. In addition, extending the proposed scheme to a time-varying environment is another interesting topic for future work.

\end{document}